\documentclass[11pt,british]{extarticle}
\usepackage{mathptmx}
\usepackage[T1]{fontenc}
\usepackage[a4paper]{geometry}
\usepackage{color}
\usepackage{babel}
\usepackage{float}
\usepackage{mathrsfs}
\usepackage{mathtools}
\usepackage{amsmath}
\usepackage{amsthm}
\usepackage{amssymb}
\usepackage{graphicx}
\usepackage[numbers]{natbib}
\usepackage[unicode=true,pdfusetitle,
 bookmarks=true,bookmarksnumbered=true,bookmarksopen=false,
 breaklinks=false,pdfborder={0 0 1},backref=false,colorlinks=true]
 {hyperref}
\hypersetup{
 pdfborderstyle=,linkcolor=black,citecolor=blue,urlcolor=black,filecolor=blue,pdfpagelayout=OneColumn,pdfnewwindow=true,pdfstartview=XYZ,plainpages=false}

\makeatletter
\numberwithin{equation}{section}
\theoremstyle{plain}
\newtheorem{thm}{\protect\theoremname}[section]
\theoremstyle{plain}

\theoremstyle{definition}

\@ifundefined{date}{}{\date{}}
\usepackage{babel}
\usepackage{caption}
\usepackage[nottoc]{tocbibind}
\allowdisplaybreaks[4]
\usepackage{dsfont}
\DeclareMathAlphabet{\mathcal}{OMS}{cmsy}{m}{n}

\providecommand{\lemmaname}{Lemma}

\providecommand{\theoremname}{Theorem}

\makeatother

\providecommand{\definitionname}{Definition}
\providecommand{\lemmaname}{Lemma}
\providecommand{\theoremname}{Theorem}

\begin{document}

\title{Random large eddy simulation for 3-dimensional incompressible
viscous flows}

\author{Zihao Guo\thanks{Zhongtai Securities Institute for Financial Studies, Shandong University, Jinan, China, 250100, Email:
\protect\href{mailto:gzhsdu@mail.sdu.edu.cn}{gzhsdu@mail.sdu.edu.cn}} \ and \ Zhongmin Qian\thanks{Mathematical Institute, University of Oxford, Oxford, United Kingdom, OX2 6GG, and Oxford Suzhou Centre for Advanced Research, Suzhou, China.  Email:
\protect\href{mailto:qianz@maths.ox.ac.uk}{qianz@maths.ox.ac.uk}}}
\maketitle
\begin{abstract}

We develop a numerical method for simulation of incompressible viscous flows by integrating the technology of random vortex method with the core idea of Large Eddy Simulation (LES). Specifically, we utilize the filtering method in LES, interpreted as spatial averaging, along with the integral representation theorem for parabolic equations, to achieve a closure scheme which may be used for calculating solutions of Navier-Stokes equations. This approach circumvents the challenge associated with handling the non-locally integrable 3-dimensional integral kernel in the random vortex method and facilitates the computation of numerical solutions for flow systems via Monte-Carlo method. Numerical simulations are carried out for both  laminar and turbulent flows, demonstrating the validity and effectiveness of the method. 

\medskip{}

\medskip{}

\end{abstract}


\section{Introduction}

The idea of utilising  probabilistic approach
is not new in Computational Fluid Dynamics (CFD), although, to the
best knowledge of the present authors, it seems that methods such
as Monte-Carlo simulation, particle system method of turbulence and other
stochastic methods have not been used extensively in industry. While these
methods are quite stimulating and there is a high potential of being useful tools
in the study of fluid flows including turbulent flow simulation. Among
them, the random vortex method (cf. \citep{Chorin1973}, \citep{CottetKoumoutsakos2000} and \citep{Majda and Bertozzi 2002} for details) developed over the last three decades
stands out as an important tool, whose successful applications in
simulations of incompressible flows however are limited to special
flows, mainly are restricted to two dimensional flows. The random
vortex method for three dimensional (3D) flows has been developed recently
in \citep{Qian2022} by using a Feynman-Kac type formula, revealing the
difficulty caused by the non-linear stretching term which does not
appear in two dimensional flows. Technically, the main problem in
the random vortex method for three dimensional incompressible flows
lies in the fact the singular integral kernel appearing in the stochastic
integral representation is no longer locally integrable, which certainly
induces instability of corresponding numerical schemes. Let us recall the key idea
in random vortex methods. Let $u(x,t)$
and $p(x,t)$ denote the velocity and the pressure respectively. Then
the motion equation of the flow is the Navier-Stokes equations (cf.
\citep{Landau1987}):

\[
\partial_{t}u+(u\cdot\nabla)u=\nu\Delta u-\nabla p+F
\]
and

\[
\nabla\cdot u=0,
\]
where $\nu>0$ is the kinematic viscosity, and $F$ is an external force applying
to the flow. The key observation is that the velocity $u$ can be
recovered from the vorticity $\omega=\nabla\wedge u$ through the
Biot-Savart law by integrating the vorticity against a singular integral kernel $K(y,x)$.
In fact, $K$ is the gradient of the Green function, and therefore
$K$ is a nice kernel for integration. While for three dimensional flows,
a further derivative of the Biot-Savart kernel $K$ is required, which
is not locally integrable unfortunately. This difficulty originates from the
random vortex method which aims to  calculating numerically
the vorticity $\omega$ first based on the vorticity transport equation
\[
\partial_{t}\omega+(u\cdot\nabla)\omega=\nu\Delta\omega+(\omega\cdot\nabla)u+\nabla\wedge F
\]
and the velocity $u$ is recovered from $\omega$. Let us recall the
technical steps briefly for two dimensional flows and the external
vorticity $\nabla\wedge F=0$ for motivating our approach in the present
paper. For this case, the non-linear stretching term $(\omega\cdot\nabla)u$
vanishes identically. Let $X$ be Brownian fluid particles $X$ with
velocity $u(x,t)$. That is,  $X$ is determined by stochastic differential equation 
\begin{equation}
\textrm{d}X_{t}=u(X_{t},t)\textrm{d}t+\sqrt{2\nu}\textrm{d}B_t, \ \ X_0=\xi. \label{SDE-u}
\end{equation}
For avoiding confusion, we use  $X^{\xi}$ to denote Brownian fluid particles issued from
a location $\xi$. Since $\omega$ is
a solution of the vorticity transport equation, an easy exercise shows that
\[
\omega(y,t)=\int h(0,\xi;t,y)u_{0}(\xi)\textrm{d}\xi,
\]
where $h(s,x;t,y)$ (for $s<t$) is the transition probability function
of the Taylor diffusion $X$, cf. \citep{Stroock1979}. Together with the Biot-Savart law and
the Fubini theorem we may represent the velocity
\[
u(x,t)=\int K(y,x)\omega(y,t)\textrm{d}y=\int\left[\int K(y,x)h(0,\xi;t,y)\textrm{d}y\right]u_{0}(\xi)\textrm{d}\xi.
\]
Now the integral with respect to $y$, that is, $\int K(y,x)h(0,\xi;t,y)\textrm{d}y$
is the average $\mathbb{E}\left[K(X_{t}^{\xi},x)\right]$, so that
$u(x,t)=\int\mathbb{E}\left[K(X_{t}^{\eta},x)\right]u_{0}(\eta)\textrm{d}\eta$,
which allows to reformulate the equation for $X$ as a stochastic
differential equation of Vlasov-McKean type
\[
\textrm{d}X_{t}^{\xi}=\left(\int\left.\mathbb{E}\left[K(X_{t}^{\eta},x)\right]u_{0}(\eta)\textrm{d}\eta\right|_{x=X_{t}^{\xi}}\right)\textrm{d}t+\sqrt{2\nu}\textrm{d}B,\quad X_{0}^{\xi}=\xi,
\]
and therefore numerical schemes may be developed accordingly for calculating
two dimensional flows. The idea similarly applies to three dimensional
flows but unfortunately the resulting stochastic integral equation involves the gradient of $K$, cf. 
\citep{Li2023}, \citep{Qian2022} and \citep{Qian2022_2} for the details
about random vortex method and integral representation theorem.

There are mainly two difficulties in numerically solving fluid dynamic
equations of turbulent flows, which are non-linear characteristics
of the Navier-Stokes equation and continuously changing complex vortices
of different scales. In engineering applications, people are more
concerned about the time-averaged effect of turbulent motion, so some
of the commonly used turbulence models are based on Reynolds time
averaging, called Reynolds Average Numerical Simulation (RANS). RANS
smooths out some tiny details of turbulent motion, and the model has
many artificial settings. Therefore, RANS models have limited simulation
capabilities for complex and delicate turbulent structure such as
the separation phenomena of flow past body. With the increasing of
the computing power and computer memory capacity, some advanced research
institutions solve the Navier-Stokes equations without any form of
simplification, but use extremely fine grids for directly numerically
computing solutions of the Navier-Stokes equations. This kind of approach may be called
Direct Numerical Simulation (DNS). However, DNS are still computational
expanse to implement. The Large Eddy Simulation (LES) method, which
is between DNS and the RANS method, has gradually emerged in CFD community and has developed into a promising method
for numerically solving turbulent flows because it is more sophisticated
than the RANS and can be implemented on a conventional computer.

LES is considered as a useful tool for numerically
simulating turbulent  flows (cf. \citep{Lesieur2005} and \citep{Volker2004}).
The principal idea in LES is to reduce the computational cost by ignoring
the eddies in small length scales, which are very expensive in numerical
experiments, via filtering of the Navier--Stokes equations. Such
filtering, which can be viewed as spatial-averaging, effectively removes
the information of eddies in small-scale from the numerical solution, for original ideas, cf.
 \citep{Deardorff1970} and \citep{Reynolds1895}. The filtered velocity
$\tilde{u}(x,t)$ can be represented as 

\[
\tilde{u}(x,t)=\int_{\mathbb{R}^{3}}\chi(x-x')u(x',t)\mathrm{d}x'
\]
where $x,x'\in\mathbb{R}^{3}$, $\chi(x,t)$ is the filtering function
including top-hat function, Gaussian function and so on, which satisfies that
$\intop\chi(x)\mathrm{d}x=1$. So the filtered Navier-Stokes equation
becomes 
\[
\partial_{t}\tilde{u}+\widetilde{(u\cdot\nabla)u}=\nu\Delta\tilde{u}-\nabla\tilde{p}+\tilde{F}
\]
and
\[
\nabla\cdot\tilde{u}=0,
\]
Because $\widetilde{(u\cdot\nabla)u}$ can not be simplified as $(\tilde{u}\cdot\nabla)\tilde{u}$,
so researchers introduced sub-grid-scale stress $\tau$, defined as 
\[
\nabla\cdot\tau=(\tilde{u}\cdot\nabla)\tilde{u}-\widetilde{(u\cdot\nabla)u}
\]
The filtered Navier-Stokes equation can be rewritten as 
\[
\partial_{t}\tilde{u}+(\tilde{u}\cdot\nabla)\tilde{u}=\nu\Delta\tilde{u}-\nabla\tilde{p}+\tilde{F}-\nabla\cdot\tau
\]

The tensor $\tau$ measures the effect of small eddies on the flow system which
must be modelled because its information is not irrelevant. The influence
of small-scale eddies on the equation of motion is described by some
other models (cf. \citep{Piomelli2002}, \citep{Smagorinsky1963} and
\citep{Pitsch2006}), but it is not the study target of our paper.
Through LES, the filtered equations focus on analysing
large scale eddies, directly capturing the most important turbulent
structure and dynamical information, while approximating small-scale
eddies. This method effectively reduces the computational complexity
and retains the main physical properties. 

Random large eddy simulation
method we are going to develop in the present paper combines the core idea of filtering
in LES with the integral representation of parabolic
equation in random vortex method,
and subsequently obtain the numerical method of incompressible viscous
flows based on it. Using the stochastic integral representation of
parabolic equation and basic idea of LES, we are
able to close the solution of Navier-Stokes equation, which makes
it possible for us to simulate incompressible flows via Monte-Carlo
method. 

The paper is organised as follows. In Section 2, we propose
the random large eddy simulation method based on stochastic integral
representation and LES. Then, numerical schemes
and simulation experiments are presented in Section 3. Finally, we
introduce conditional law duality and integral representation theorem
in Appendix, which is the theoretical foundation of the random large
eddy simulation method established in Section 2. 

\section{Random large eddy simulation method}

The present work aims to implement the ideas of random vortex method
described in the introduction directly to the velocity rather than
through the vorticity equation, so that the method developed in the
present paper works for three dimensional (turbulent) flows as well.
Firstly, for an incompressible flow, the divergence of the velocity
vanishes: $\nabla\cdot u=0$, which implies that the probability transition
function $p(s,x;t,y)$ is the fundamental solution for the backward
parabolic operator $\partial_{t}+\nu\triangle+u\cdot\nabla$, so it
is also the fundamental solution of the forward parabolic operator
$\partial_{t}-\nu\triangle+u\cdot\nabla$. Therefore $u(x,t)$ has the following implicit integral representation 
\begin{equation}
u(x,t)=\int_{\mathbb{R}^3} p(0,\eta;t,x)u_{0}(\eta)\textrm{d}\eta+\int_{0}^{t}\int_{\mathbb{R}^3} h(s,\eta;t,x)\left(-\nabla h(\eta,s)+F(\eta,s)\right)\textrm{d}\eta\textrm{d}s\label{u-rep1}
\end{equation}
where $\nabla p$ can be calculated by using the Biot-Savart law.
In fact, since $\nabla\cdot u=0$, the pressure $p(x,t)$ at every
instance $t$ satisfies the Poisson equation

\[
\Delta p=-\sum_{i,j=1}^{3}\frac{\partial u^{j}}{\partial x^{i}}\frac{\partial u^{i}}{\partial x^{j}}+\nabla\cdot F
\]
so that, according to Green formula,

\[
p(x,t)=\int_{\mathbb{R}^{3}}G_{3}(y,x)\left.\left(-\sum_{i,j=1}^{3}\frac{\partial u^{j}}{\partial x^{i}}\frac{\partial u^{i}}{\partial x^{j}}+\nabla\cdot F\right)\right|_{(y,t)}\mathrm{d}y
\]
where $G_{3}(y,x)=-\frac{1}{4\pi}\frac{1}{\left|y-x\right|}$ is the
Green function on $\mathbb{R}^{3}$. Differentiating under integration
(which we assume is legible) to deduce that
\begin{equation}
\nabla p(x,t)=\int_{\mathbb{R}^{3}}K_{3}(y,x)\left.\left(\sum_{i,j=1}^{3}\frac{\partial u^{j}}{\partial x^{i}}\frac{\partial u^{i}}{\partial x^{j}}-\nabla\cdot F\right)\right|_{(y,t)}\mathrm{d}y,\label{grad-p-01}
\end{equation}
where
\[
K_{3}(y,x)=-\nabla_{x}G_{3}(y,x)=\frac{1}{4\pi}\frac{y-x}{\left|y-x\right|^{3}}\quad\textrm{ for }y\neq x
\]
is the Biot-Savart kernel. The representation (\ref{grad-p-01}) shall
be used to update the pressure $p(x,t)$ in solving numerically the
velocity $u(x,t)$ via the Navier-Stokes equations.

For simplicity we set $g=-\nabla p+F$. For every $\xi$ and instance
$t\geq0$, $X_{t}^{\xi}$ denotes the position in the state space
of the Brownian fluid particles with velocity $u(x,t)$ started from
$\xi$, and $X:(t,\xi)\rightarrow X_{t}^{\xi}$ defines a random field.
In terms of the law of random field $X$, (\ref{u-rep1}) can be written
formally as
\begin{equation}
u(x,t)=\int_{\mathbb{R}^3}\mathbb{E}\left[\delta_{x}(X_{t}^{\eta})\right]u_{0}(\eta)\textrm{d}\eta+\int_{0}^{t}\int_{\mathbb{R}^3}\mathbb{E}\left[\left.\delta_{x}(X_{t}^{\eta})\right|X_{s}^{\eta}\right]g(\eta,s)\textrm{d}\eta\textrm{d}s\label{u-rep1-1}
\end{equation}
by using the formal symbol $h(0,\eta;t,x)=\mathbb{P}\left[X_{t}^{\eta}=x\right]$.
The problem here is that of course $\delta_{x}(X_{t}^{\eta})$ is
a generalised Wiener functional so that it is difficult to calculate
it numerically. Also the conditional average $\mathbb{E}\left[\left.\delta_{x}(X_{t}^{\eta})\right|X_{s}^{\eta}\right]$
is expensive to compute. To overcome these difficulties, we utilise
a key idea from LES. In LES,
one calculates local averaged velocity which shall give the global
structure of the flow (in particular for turbulent flows). We shall
adopt the LES approach, avoiding the modelling of the error term caused
by taking local average of the velocity.

To realise this scheme, we shall use a stochastic integral representation
for solutions of a linear parabolic equation. Since $u$ is a solution to 
\begin{equation}
\left(\nu\Delta-u\cdot\nabla-\partial_{t}\right)u+g=0,\label{eq:parabolic1-1}
\end{equation}
according to the representation theorem (\ref{Thm:representation})
to be established in Appendix
\begin{equation}
u(x,t)=\int_{\mathbb{R}^{3}}h(0,\eta;t,x)u_{0}(\eta)\mathrm{d}\eta+\int_{0}^{t}\int_{\mathbb{R}^{3}}\mathbb{E}\left[g(X_{s}^{\eta},s)\left|X_{t}^{\eta}=x\right.\right]h(0,\eta,t,x)\mathrm{d}\eta\mathrm{d}s.\label{rep-u2}
\end{equation}
Comparing with (\ref{u-rep1-1}), the difference seems not so significant, and indeed the two representations are equivalent due to the assumption
that $\nabla\cdot u=0$, while (\ref{rep-u2}) involves only Brownian
fluid particles started at the same instance, which in fact greatly
reduces the computational cost.

By borrowing the key idea in LES, a filter function $\chi$ with a
compact support is applied for calculating the local average of the
velocity, denoted by $\tilde{u}(x,t)$. More precisely, 

\begin{equation}
\tilde{u}(x,t)=\int_{\mathbb{R}^{3}}\chi(\xi-x)u(\xi,t)\mathrm{d}\xi\label{filter_u}
\end{equation}
where $x\in\mathbb{R}^{3}$. In the present paper, we choose a parameter
$s>0$ as the mesh size, the filter function used in the paper is
Gaussian kernel $\chi(x)=\left(\frac{6}{\pi s^{2}}\right)^{\frac{3}{2}}\textrm{e}^{-\frac{6x^{2}}{s^{2}}}$
for $x\in\mathbb{R}^{3}$. 

Combining (\ref{rep-u2}) and (\ref{filter_u}), we may easily obtain
the following functional integral representation
\begin{align}
\tilde{u}(x,t) & =\int_{\mathbb{R}^{3}}\mathbb{E}\left[\chi(X_{t}^{\eta}-x)\right]u_{0}(\eta)\mathrm{d}\eta\nonumber \\
 & +\int_{0}^{t}\int_{\mathbb{R}^{3}}\mathbb{E}\left[\chi(X_{t}^{\eta}-x)g(X_{s}^{\eta},s)\right]\mathrm{d}\eta\mathrm{d}s\label{eq:rep_u}
\end{align}
for $x\in\mathbb{R}^{3}$, where the Brownian fluid particles are
defined in terms of the stochastic differential equation (\ref{SDE-u}).
In the approach of LES, one derives the evolution equation for the
locally averaged velocity $\tilde{u}(x,t)$, namely the Navier-Stokes
equations with the error term, and LES is then implemented by modelling
the error. Taking advantage of the integral representation (\ref{eq:rep_u}),
which though is an implicit representation, we may devise a natural
closure scheme without further modelling for the evolution of the
locally averaged velocity $\tilde{u}(x,t)$. 

More precisely, given a filter function $\chi$, we run Brownian fluid
particles with the locally averaged velocity $\hat{u}(x,t)$ defined by stochastic differential equation

\begin{equation}
\mathrm{d}\hat{X}_{t}^{\eta}=\hat{u}(\hat{X}_{t}^{\eta},t)\mathrm{d}t+\sqrt{2\nu}\mathrm{d}B_{t},\quad\hat{X}_{0}^{\eta}=\eta,\label{eq:rep_X}
\end{equation}
\begin{align}
\hat{u}(x,t)= & \int_{\mathbb{R}^{3}}\mathbb{E}\left[\chi(\hat{X}_{t}^{\eta}-x)\right]u_{0}(\eta)\mathrm{d}\eta\nonumber \\
 & +\int_{0}^{t}\int_{\mathbb{R}^{3}}\mathbb{E}\left[\chi(\hat{X}_{t}^{\eta}-x)\hat{g}(\hat{X}_{s}^{\eta},s)\right]\mathrm{d}\eta\mathrm{d}s,\label{hat-u}
\end{align}
where 
\begin{equation}
\hat{g}(x,t)=-\nabla\hat{p}(x,t)+F(x,t),\label{hat-g}
\end{equation}
and

\begin{equation}
\nabla\hat{p}(x,t)=\int_{\mathbb{R}^{3}}K_{3}(\hat{X}_{t}^{\eta},x)\left(\sum_{i,j=1}^{3}\frac{\partial\hat{u}^{j}}{\partial x_{i}}\frac{\partial\hat{u}^{i}}{\partial x_{j}}-\nabla\cdot F\right)\bigg|_{(\hat{X}_{t}^{\eta},t)}\mathrm{d}\eta.\label{eq:rep_g}
\end{equation}

The equations (\ref{eq:rep_X}, \ref{hat-u}, \ref{hat-g}, \ref{eq:rep_g})
compose a closed system which allows to calculate numerically approximate
locally averaged velocity of an incompressible fluid flow. 

If we utilise a Gaussian filter function $\chi$ , to update the approximate
pressure $\hat{p}(x,t)$, terms $\frac{\partial\hat{u}}{\partial x}$
in $\nabla\hat{p}$ can be calculated simply by differentiating $\hat{u}(x,t)$
via (\ref{hat-u}). 

\section{Numerical schemes and numerical experiments}

In this section we present several numerical experiments based on
the scheme defined by (\ref{eq:rep_X}, \ref{hat-u}, \ref{hat-g},
\ref{eq:rep_g}). For simplicity, we drop the hat script, and therefore
define the following system
\begin{equation}
\mathrm{d}X_{t}^{\eta}=u(X_{t}^{\eta},t)\mathrm{d}t+\sqrt{2\nu}\mathrm{d}B_{t},\quad X_{0}^{\eta}=\eta\label{f-X-1}
\end{equation}
where $B$ is a Brownian motion on some probability space, 
\begin{equation}
u(x,t)  =\int_{\mathbb{R}^{3}}\mathbb{E}\left[\chi(X_{t}^{\eta}-x)\right]u_{0}(\eta)\mathrm{d}\eta
+\int_{0}^{t}\int_{\mathbb{R}^{3}}\mathbb{E}\left[\chi(X_{t}^{\eta}-x)g(X_{s}^{\eta},s)\right]\mathrm{d}\eta\mathrm{d}s,
\label{f-u-1}
\end{equation}
\begin{equation}
g(x,t)=-\nabla p(x,t)+F(x,t)\label{f-g-1}
\end{equation}
and the pressure gradient is updated by
\begin{equation}
\nabla p(x,t)=\int_{\mathbb{R}^{3}}K_{3}(X_{t}^{\eta},x)\left(\sum_{i,j=1}^{3}\frac{\partial u^{j}}{\partial x_{i}}\frac{\partial u^{i}}{\partial x_{j}}-\nabla\cdot F\right)\bigg|_{(X_{t}^{\eta},t)}\mathrm{d}\eta.\label{f-p-1}
\end{equation}
In this scheme, the initial velocity $u_{0}$ and the external force
$F$ are given data. The whole scheme, which may be called the random
LES, is determined by a filter function $\chi$ which must be chosen
according to the mesh size when implement the simulation. 

The initial velocity $u(x,0)$ is given, so
the initial pressure $p(x,0)$ can be calculated directly, we choose
external force $F$ to be a constant so that $\nabla\cdot F=0$. 

Set mesh size $s>0$, time step $\delta>0$ and kinematic viscosity
$\nu>0$. For $i_{1},i_{2},i_{3}\in\mathbb{Z}$, denote $x^{i_{1,}i_{2},i_{3}}=(i_{1},i_{2},i_{3})s$,
$u^{i_{1},i_{2},i_{3}}=u(x^{i_{1,}i_{2},i_{3}},0)$. In numerical
scheme, we drop the expectation and use one-copy of Brownian particles.
We discrete the stochastic differential equation (\ref{f-X-1}) following
Euler scheme: for $t_{i}=i\delta,i=0,1,2,\cdots$.

\[
X_{t_{k}}^{i_{1},i_{2},i_{3}}=X_{t_{k-1}}^{i_{1},i_{2},i_{3}}+\delta u(X_{t_{k-1}}^{i_{1},i_{2},i_{3}},t_{k-1})+\sqrt{2\nu}(B_{t_{k}}-B_{t_{k-1}}),\quad X_{0}^{i_{1},i_{2},i_{3}}=x^{i_{1,}i_{2},i_{3}}
\]
where $X_{t_{k}}^{i_{1},i_{2},i_{3}}$ denotes $X_{t_{k}}^{x^{i_{1},i_{2},i_{3}}}$
for simplicity. Then integral representations in (\ref{f-u-1}), and
(\ref{f-p-1}) can be discretised as follows:
\[
u(x,t_{k})  =\sum_{i_{1},i_{2},i_{3}}s^{3}\chi(X_{t_{k}}^{i_{1},i_{2},i_{3}}-x)u^{i_{1},i_{2},i_{3}}
+\sum_{i_{1},i_{2},i_{3}}\sum_{j=1}^{k}s^{3}\delta\chi(X_{t_{k}}^{i_{1},i_{2},i_{3}}-x)g(X_{t_{j-1}}^{i_{1},i_{2},i_{3}},t_{j-1})
\]
and
\[
g(x,t_{k})=-\sum_{i_{1},i_{2},i_{3}}s^{3}K(X_{t_{k}}^{i_{1},i_{2},i_{3}},x)\sum_{i,j=1}^{3}A_{j}^{i}(X_{t_{k}}^{i_{1},i_{2},i_{3}},t_{k}) A_{i}^{j}(X_{t_{k}}^{i_{1},i_{2},i_{3}},t_{k})+F,
\]
where
\[
A(x,t_{k})=\sum_{i_{1},i_{2},i_{3}}s^{3}\nabla_{x}\chi(X_{t_{k}}^{i_{1},i_{2},i_{3}}-x)u^{i_{1},i_{2},i_{3}}
+\sum_{i_{1},i_{2},i_{3}}\sum_{j=1}^{k}s^{3}\delta\nabla_{x}\chi(X_{t_{k}}^{i_{1},i_{2},i_{3}}-x)g(X_{t_{j-1}}^{i_{1},i_{2},i_{3}},t_{j-1}).
\]

We next carry out several numerical experiments by using the previous scheme.

\subsection{Numerical experiments -- laminar flows}

By a laminar flow we mean a viscous flow with a small Reynolds number,  where the Reynolds number
is defined by $Re=\frac{U_{0}L}{\nu}$, where $U_{0}$ is the main stream
velocity and $L$ is the length scale. In our experiment,  $\nu=0.15$, $Re=800$,
length scale $L=2\pi$, so that $U_{0}=\frac{\nu}{L}Re=23.8$. The mesh
size $s=\frac{2\pi}{20}\thicksim L\sqrt{\frac{1}{Re}}$. and the time
step $\delta=0.001$. The numerical experiment is demonstrated at times
$t=0.3,t=0.6,t=0.9$. We set the initial velocity to be of the form
$u(x,0)=(U_{0}\sin(2x_{1}),U_{0}\cos(2x_{1}),0)$, and set force $F=(10,10,-9.81)$.
The velocity field and the vorticity field are shown in Figure \ref{l_u} and Figure \ref{l_w}. 

\begin{figure}[H]
\includegraphics[width=1\textwidth]{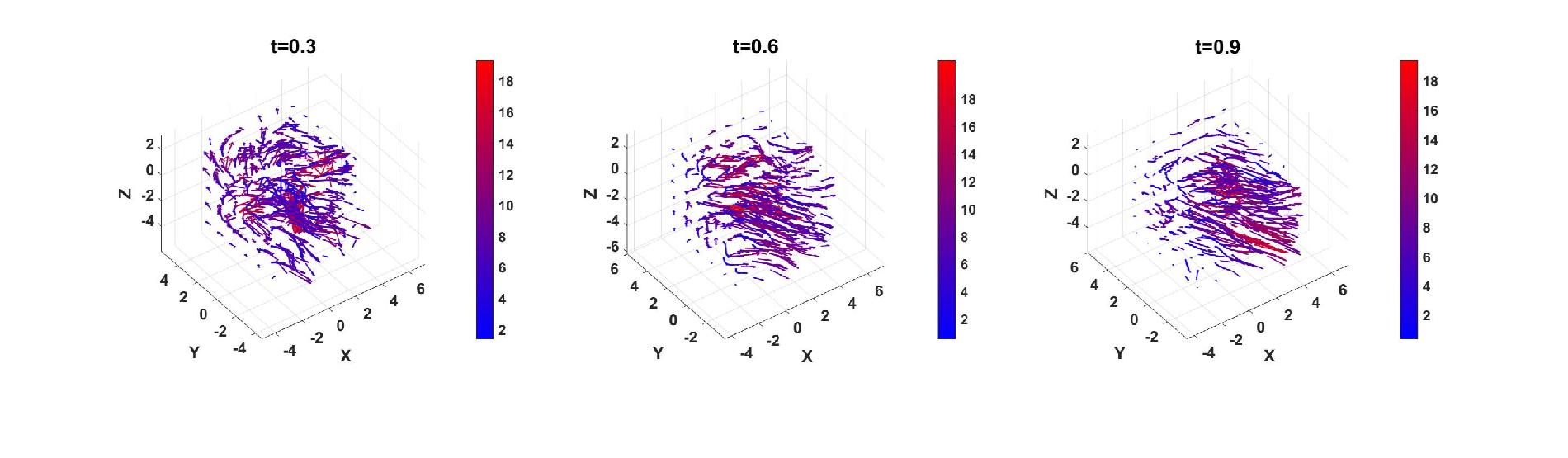}

\caption{\foreignlanguage{english}{Velocity fields of incompressible viscous flows on $\mathbb{R}^{3}$}}
\label{l_u}
\end{figure}

\begin{figure}[H]
\includegraphics[width=1\textwidth]{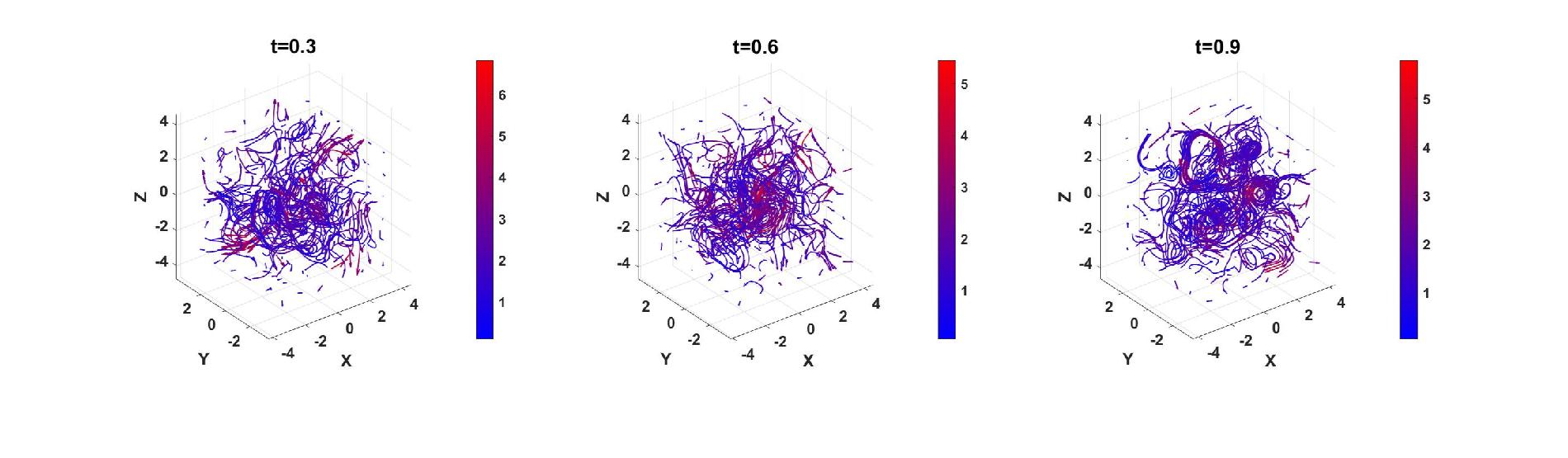}

\caption{\foreignlanguage{english}{Vorticity fields of incompressible viscous flows flows on $\mathbb{R}^{3}$}}
\label{l_w}
\end{figure}

Figure \ref{l_u} and Figure \ref{l_w} show a laminar flow where the velocity varies smoothly, while the vorticity field changes relatively fast, which is consistent with the real flow in question. To view the flows clearly, we also print sections of the vector fields 
with the plane $z=0$, plane $x=0$, and the plane $y=0$ in Figure \ref{l_z}, Figure \ref{l_x}, 
Figure \ref{l_y} respectively.

\begin{figure}[H]
\includegraphics[width=1\textwidth]{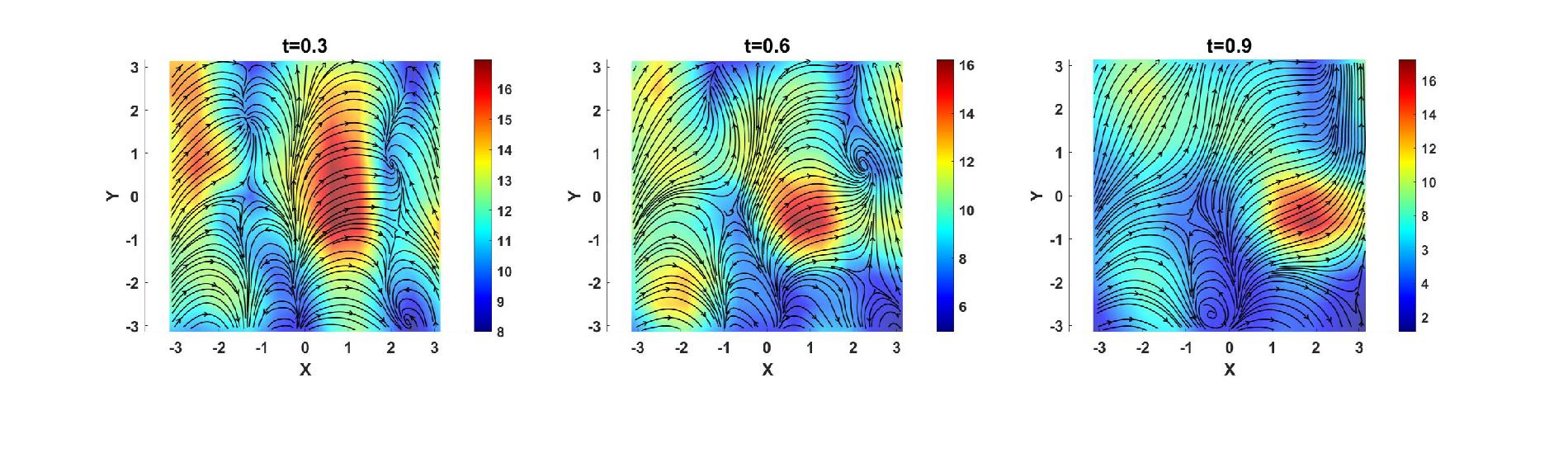}

\caption{\foreignlanguage{english}{Section Velocity fields of incompressible viscous flows on plane $z=0$}}
\label{l_z}
\end{figure}
As shown in Figure \ref{l_z}, because of the initial conditions, there are more vortices and complex flows occurring on the plane $z=0$. The external force $F$ applying to this laminar flows system is not strong enough, so the fluid changes are not particularly large. 

\begin{figure}[H]
\includegraphics[width=1\textwidth]{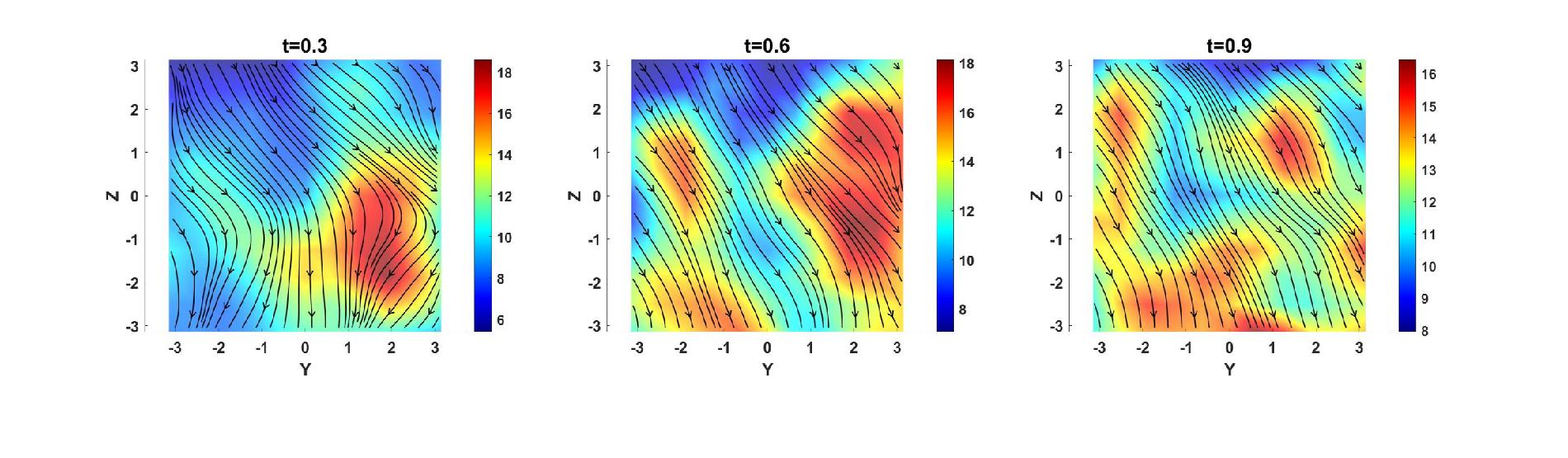}

\caption{\foreignlanguage{english}{Section Velocity fields of incompressible viscous flows on plane $x=0$}}
\label{l_x}
\end{figure}

\begin{figure}[H]
\includegraphics[width=1\textwidth]{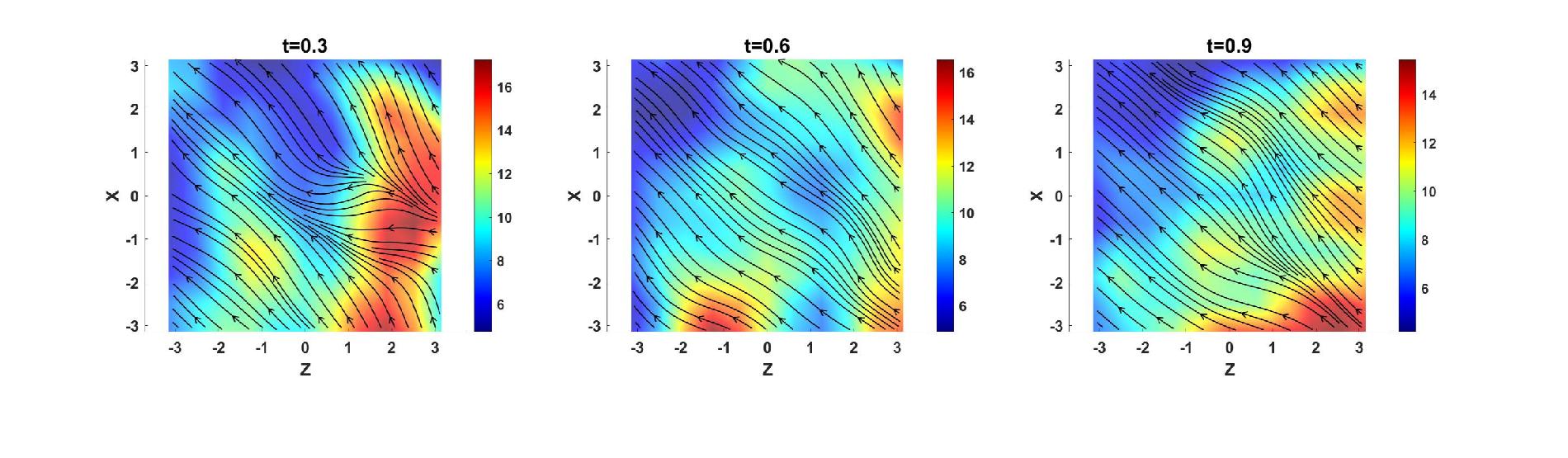}

\caption{\foreignlanguage{english}{Section Velocity fields of incompressible viscous flows on plane $y=0$}}
\label{l_y}
\end{figure}

The initial vorticity is apparent  only in plane $z=C$ ($C$ is constant),
it is clear that more obvious vortexes and changes appear in the horizontal
section because the force $F$ applying into the flows system is constant.
So, in plane $x=C$ and plane $y=C$, the fluid shows a tendency
to flow normally, flows in the direction of the external force.

\subsection{Turbulence flows}

We present a numerical experiment showing a 3D turbulent flow, a flow with large Reynolds number. In the numerical
experiment,  $\nu=0.15$, $Re=4500$, length scale $L=2\pi$,
so that $U_{0}=107.5$. The mesh size $s=\frac{2\pi}{40}$, and the time step
$\delta=0.001$. The numerical experiment results are demonstrated at times
$t=0.05,t=0.1,t=0.15$. We set the initial velocity to be of the form
$u(x,0)=(U_{0}\sin(2x_{1}),U_{0}\cos(2x_{1}),0)$, set force $F=(50,50,-9.81)$.
The velocity field and vorticity field are shown in Figure \ref{t_u} and Figure \ref{t_w}. 

\begin{figure}[H]
\includegraphics[width=1\textwidth]{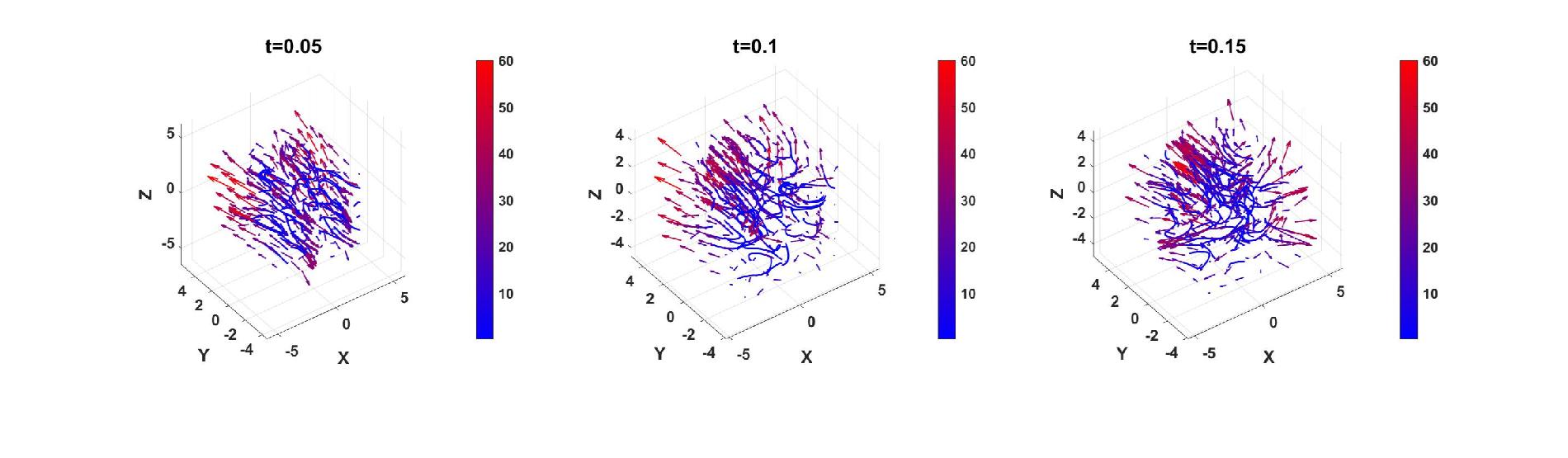}

\caption{\foreignlanguage{english}{Velocity fields of incompressible viscous flows on $\mathbb{R}^{3}$}}
\label{t_u}
\end{figure}

\begin{figure}[H]
\includegraphics[width=1\textwidth]{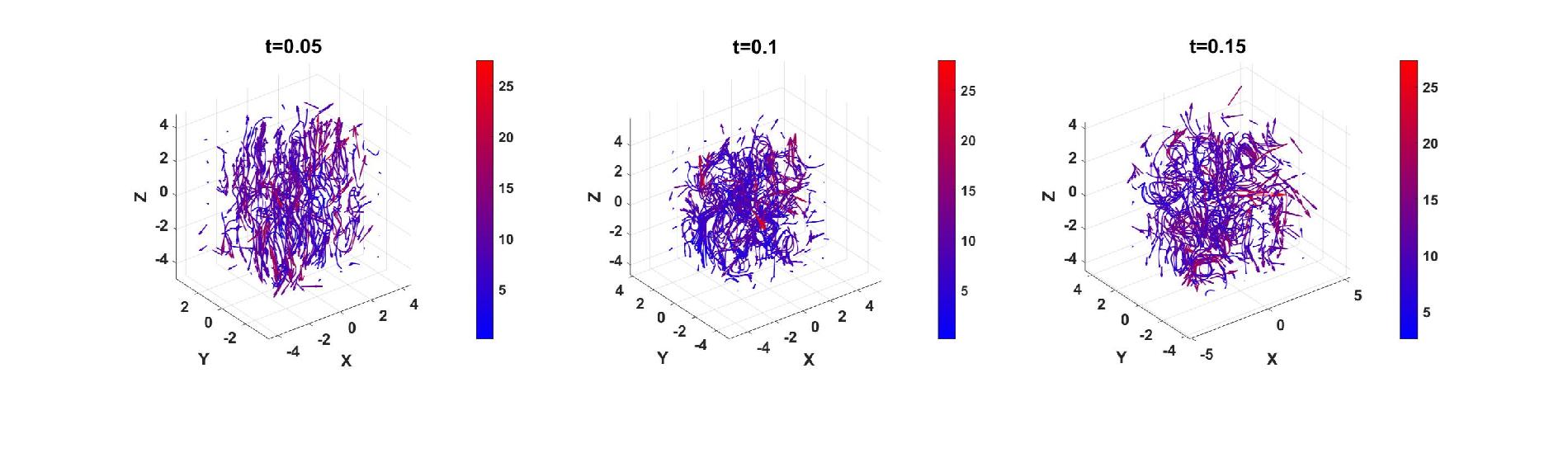}

\caption{\foreignlanguage{english}{Vorticity fields of incompressible viscous flows flows on $\mathbb{R}^{3}$}}
\label{t_w}
\end{figure}

In turbulence flows, given the initial velocity, the vorticity is larger and the vorticity field is more chaotic under larger velocity fields and external forces. To view the flows clearly, we also print the vector fields of section
with plane $z=0$, plane $x=0$, and plane $y=0$ in Figure \ref{t_z}, Figure \ref{t_x},
Figure \ref{t_y} respectively.

\begin{figure}[H]
\includegraphics[width=1\textwidth]{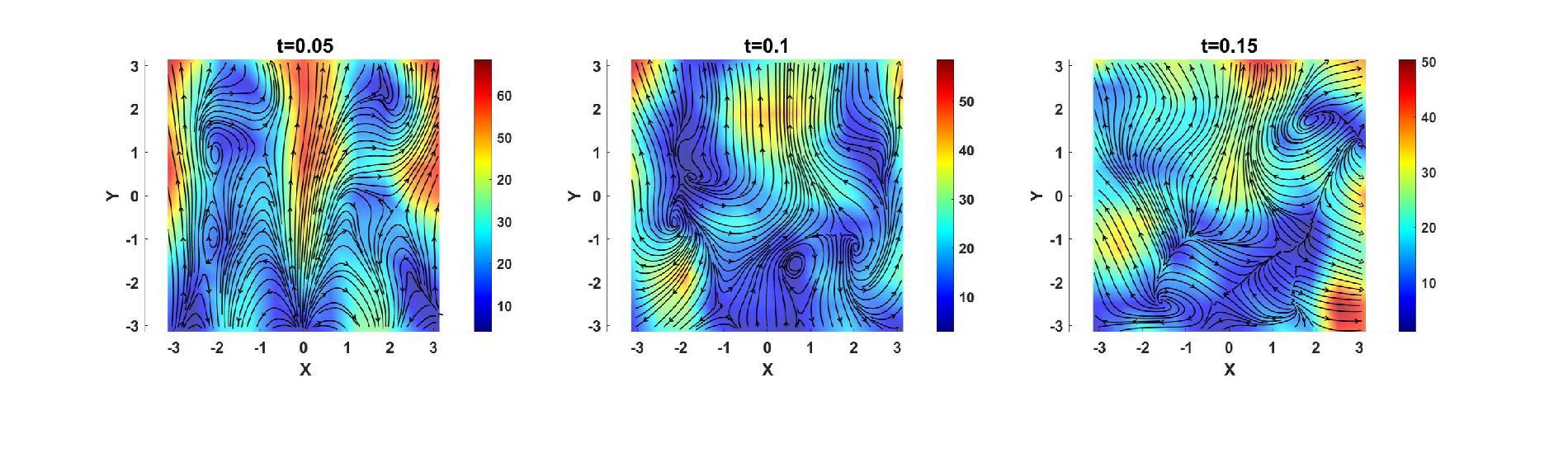}

\caption{\foreignlanguage{english}{Section Velocity fields of incompressible viscous flows on plane $z=0$}}
\label{t_z}
\end{figure}
Compared with laminar flow, more eddies and more complex flows appear on plane $z=0$ in turbulence flow. Because the initial velocity and force are relatively large, the changes in the fluid are also more significant.

\begin{figure}[H]
\includegraphics[width=1\textwidth]{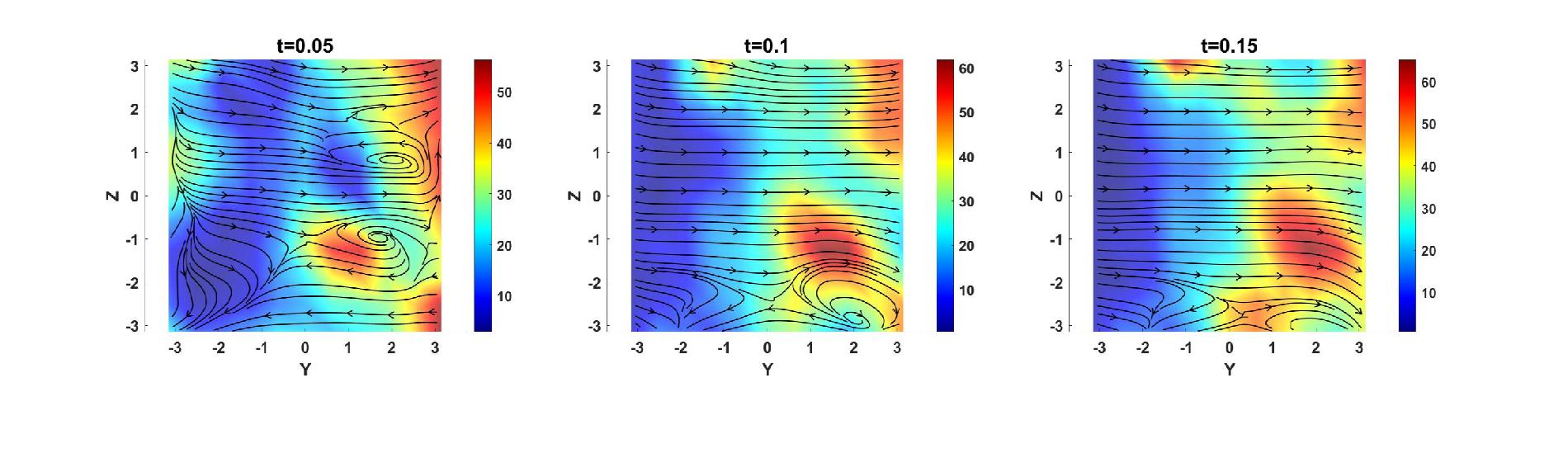}

\caption{\foreignlanguage{english}{Section Velocity fields of incompressible viscous flows on plane $x=0$}}
\label{t_x}
\end{figure}

\begin{figure}[H]
\includegraphics[width=1\textwidth]{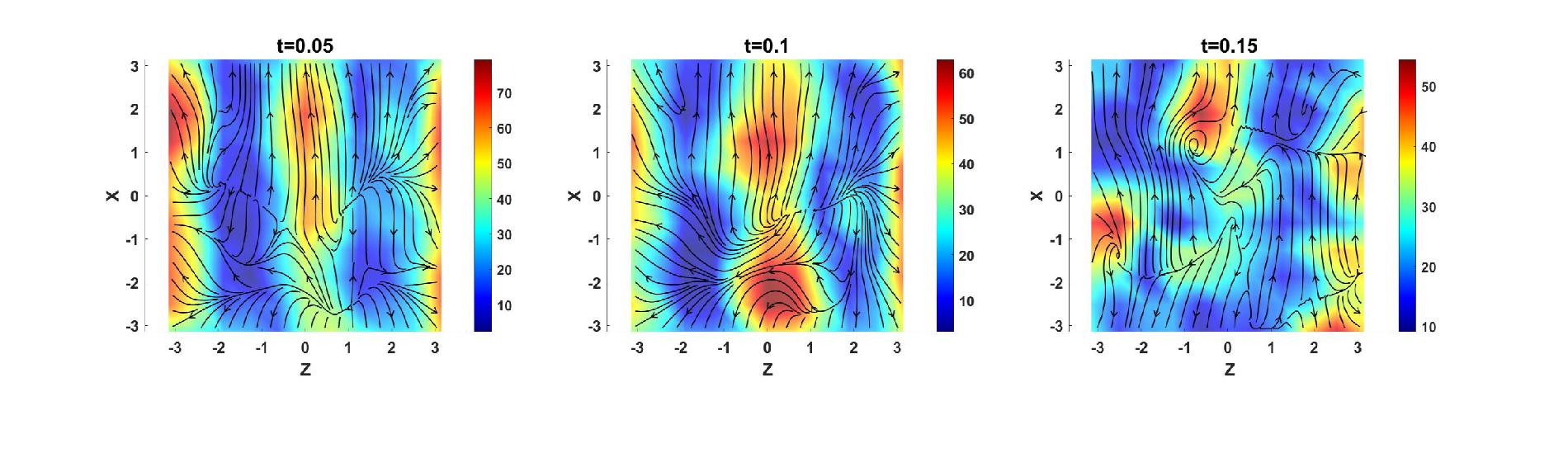}

\caption{\foreignlanguage{english}{Section Velocity fields of incompressible viscous flows on plane $y=0$}}
\label{t_y}
\end{figure}

Although the initial vorticity is only given in plane $z=C$ ($C$
is constant), there are still fewer vortexes and changes appearing
in plane $x=C$ and plane $y=C$ under turbulence case and large external
force.

\section{Appendix}

In this section, we derive the integral representation theorem which
is the key ingredient in the present work. 

Suppose $b(x,t)$ is a time-dependent divergence-free vector field
on $\mathbb{R}^{d}\times[0,\infty)$, i.e. $\nabla\cdot b=0$. Let
$L_{b}=\nu\Delta+b\cdot\nabla$, then the $L^{2}$-adjoint operator
of forward heating operator $L_{b}-\partial_{t}$ coincides with backward
heat operator $L_{-b}+\partial_{t}$ (see \citep{Friedman1964} for
details). Let $\Omega=C([0,\infty),\mathbb{R}^{d})$ be the continuous
path space in $\mathbb{R}^{d}$, define $X_{t}:\Omega\mapsto\mathbb{R}^{d}$
be the coordinate process on $\Omega$. Let $\mathscr{F}_{t}^{0}=\sigma\left\{ X_{s}:s\leq t\right\} $
be the smallest $\sigma$-algebra on $\Omega$ and $\mathscr{F}^{0}=\sigma\left\{ X_{s}:s<\infty\right\} $.
Then $\mathscr{F}^{0}=\mathcal{B}(\Omega)$ is the Borel $\sigma$-algebra
on $\Omega$ generated by the uniform convergence over any bounded
subset $[0,\infty).$ We assume $b(x,t)$ is bounded and Borel measure,
then for $\xi\in\mathbb{R}^{d},\tau\geq0$, there is a unique probability
measure $\mathbb{P}_{b}^{\xi,\tau}$ on $(\Omega,\mathscr{F}^{0})$
such that
\[
\mathbb{P}[X_{s}=\xi\;\textrm{ for }\;s\leq\tau]=1
\]
and
\[
M_{t}^{[f]}=f(X_{t},t)-f(X_{\tau},\tau)-\int_{\tau}^{t}(L_{b}f)(X_{s},s)\mathrm{d}s
\]
is a local martingale ($\tau\leq t$) under probability measure $\mathbb{P}_{b}^{\xi,\tau}$,
where $f\in C_{b}^{2,1}(\mathbb{R}^{d}\times[0,\infty))$. $\mathbb{P}_{b}^{\xi,\tau}$
is called $L_{b}$-diffusion, and $L_{b}$ is called the infinitesimal
generator of $\mathbb{P}_{b}^{\xi,\tau}$ \citep{Stroock1979}. $\mathbb{P}_{b}^{\xi,\tau}$
can be constructed as the weak solution of stochastic differential
equation
\[
\mathrm{d}X_{t}=b(X_{t},t)\mathrm{d}t+\sqrt{2\nu}\mathrm{d}B_{t},\quad X_{\tau}=\xi,
\]
where $B$ is a Brownian motion on probability space. Let $p_{b}(\tau,\xi;t,x)$
denote the transition function of of $L_{b}$-diffusion, which is
positive and continuous in all arguments for $\tau<t$ and $\xi,x\in\mathbb{R}^{d}$,
in the sense that 
\[
p_{b}(\tau,\xi;t,x)\mathrm{d}x=P(\tau,\xi,t,\mathrm{d}x)=\mathbb{P}_{b}^{\xi,\tau}[X_{t}\in\mathrm{d}x],
\]
see \citep{Stroock1979} for details. Given $T>0$, $\mathbb{P}_{b}^{\xi,\tau}[\cdot|X_{T}=\eta]$
denotes the conditional law of the $L_{b}$-diffusion. Then it has
been established in \citep{Qian2022} that the conditional law $\mathbb{P}_{b}^{\xi,0\rightarrow\eta,T}$
coincide with the conditional law $\mathbb{P}_{-b^{T}}^{\eta,0\rightarrow\xi,T}\circ\tau_{T}$
up to a time reverse at $T$, where $b^{T}(x,t)=b(x,(T-t)^{+})$.
In fact it follows immediately from the fact that, since $b$ is divergence-free, 
\[
p_{-b^{T}}(s,x;t,y)=p_{b}(T-t,y;T-s,x)
\]
for $0\leq s<t\leq T$ and $x,y\in\mathbb{R}^{d}$, cf. \citep{Qian2022}
for details. 

Let $\Psi(x,t)$ be the smooth solution of parabolic equation

\begin{equation}
\left(L_{-b}-\frac{\partial}{\partial t}\right)\Psi+f=0\quad\textrm{ in }\mathbb{R}^{d}\times[0,\infty)\label{eq:parabolic1}
\end{equation}
where $f(x,t)=$ is $C^{2,1}$-function. Assume that $b(x,t)$ is $C^{2,1}$ and bounded, then we have the following theorem.
\begin{thm}
\label{Thm:representation}Suppose $b(x,t)$ is  divergence-free, that is,  $\nabla\cdot b=0$ in the sense of distribution on $\mathbb{R}^d$. Then
\begin{align}
\Psi(\xi,T) & =\int_{\mathbb{R}^{d}}p_{b}(0,\eta;T,\xi)\Psi(\eta,0)\mathrm{d}\eta\nonumber \\
 & +\int_{0}^{T}\int_{\mathbb{R}^{d}}\mathbb{P}_b^{\eta,0}\left[f(X_{t},t)\left|X_{T}=\xi\right.\right]p_{b}(0,\eta;T,\xi)\mathrm{d}\eta\mathrm{d}t\label{eq:representation1}
\end{align}
for $\xi\in\mathbb{R}^{d}$ and $T>0$.
\end{thm}

\begin{proof}
First, we construct a diffusion associated with time-reversal vector
field $b^{T}$ by solving the following SDE:
\[
\mathrm{d}\tilde{X_{t}^{\xi}}=-b\left(\tilde{X_{t}^{\xi}},T-t\right)\mathrm{d}t+\sqrt{2\nu}\mathrm{d}B_{t},\quad\tilde{X_{0}^{\xi}}=\xi
\]
whose infinitesimal generator is $L_{-b^{T}}$. Let $Y_{t}=\Psi(\tilde{X_{t}^{\xi}},T-t)$,
then by It\^o formula and equation \ref{eq:parabolic1}
\[
Y_{t}=Y_{0}+\sqrt{2\nu}\int_{0}^{t}\nabla\Psi(\tilde{X_{s}^{\xi}},T-s)\cdot\mathrm{d}B_{s}-\int_{0}^{t}f(\tilde{X_{s}^{\xi}},T-s)\mathrm{d}s
\]
Taking expectation of both sides and $t=T$, we have
\begin{align*}
\Psi(\xi,T) & =\mathbb{E}\left[\Psi(\tilde{X_{T}^{\xi}},0)\right]+\int_{0}^{T}\mathbb{E}\left[f(\tilde{X_{t}^{\xi}},T-t)\right]\mathrm{d}s\\
 & =J_{1}+J_{2}
\end{align*}
By using conditional expectation, we may write
\begin{align*}
J_{2}(\xi,T) & =\int_{0}^{T}\int_{\mathbb{R}^{d}}\mathbb{E}\left[f(\tilde{X_{t}^{\xi}},T-t)\left|\tilde{X_{T}^{\xi}}=\eta\right.\right]\mathbb{P}\left[\tilde{X_{T}^{\xi}}\in\mathrm{d}\eta\right]\mathrm{d}t\\
 & +\int_{0}^{T}\int_{\mathbb{R}^{d}}\mathbb{E}\left[f(\tilde{X_{t}^{\xi}},T-t)\left|\tilde{X_{T}^{\xi}}=\eta\right.\right]p_{-b^{T}}(0,\xi;T,\eta)\mathrm{d}\eta\mathrm{d}t
\end{align*}
Similarly
\[
J_{1}(\xi,T)=\int_{\mathbb{R}^{d}}\Psi^{i}(\eta,0)p_{-b^{T}}(0,\xi;T,\eta)\mathrm{d}\eta.
\]
Since $\nabla\cdot b=0$, we have $p_{b}(0,\eta;T,\xi)=p_{-b^{T}}(0,\xi;T,\eta)$,
we can get
\[
J_{1}(\xi,T)=\int_{\mathbb{R}^{d}}\Psi(\eta,0)p_{b}(0,\eta;T,\xi)\mathrm{d}\eta
\]
and
\[
J_{2}(\xi,T)=\int_{0}^{T}\int_{\mathbb{R}^{d}}\mathbb{E}\left[f(\tilde{X_{t}^{\xi}},T-t)\left|\tilde{X_{T}^{\xi}}=\eta\right.\right]p_{b}(0,\eta;T,\xi)\mathrm{d}\eta\mathrm{d}t.
\]
Let $\mathbb{P}_{-b^{T}}^{\xi,0}$ be the distribution of $\tilde{X^{\xi}}$,
and $\mathbb{P}_{-b^{T}}^{\xi,0\rightarrow\eta,T}$ is the conditional
distribution under condition $\tilde{X_{T}^{\xi}}=\eta$, then we
can rewrite $J_{2}$ as
\[
J_{2}(\xi,T)=\int_{0}^{T}\int_{\mathbb{R}^{d}}\mathbb{P}_{-b^{T}}^{\xi,0\rightarrow\eta,T}\left[f(\psi(t),T-t)\right]p_{b}(0,\eta;T,\xi)\mathrm{d}\eta\mathrm{d}t
\]
which can be written as
\[
J_{2}(\xi,T)=\int_{0}^{T}\int_{\mathbb{R}^{d}}\mathbb{P}_{b}^{\eta,0\rightarrow\xi,T}\left[f(\psi(T-t),T-t)\right]p_{b}(0,\eta;T,\xi)\mathrm{d}\eta\mathrm{d}t
\]
Therefore after a change of variable for $t$ to $T-t$, 
\[
J_{2}(\xi,T)=\int_{0}^{T}\int_{\mathbb{R}^{d}}\mathbb{P}_{b}^{\eta,0\rightarrow\xi,T}\left[f(\psi(t),t)\right]p_{b}(0,\eta;T,\xi)\mathrm{d}\eta\mathrm{d}t
\]
which yields the integral representation \ref{eq:representation1}.
\end{proof}

\section*{Data Availability Statement}

No data are used in this article to support the findings of this study.

\section*{Declaration of Interests}

The authors report no conflict of interest.

\section*{Acknowledgement}
The authors would like to thank Oxford Suzhou Centre for Advanced
Research for providing the excellent computing facility.

\end{document}